\RequirePackage{amsmath}
\documentclass[runningheads,11pt]{llncs}
\usepackage[T1]{fontenc}
\usepackage{lmodern}
\usepackage{graphicx,hyperref,amsmath,amsfonts,color}

\usepackage[margin=1in,footskip=0in]{geometry}
\usepackage{fullpage}
\usepackage{algpseudocode}
\usepackage{enumitem}
\usepackage[ruled]{algorithm2e}

\DeclareMathOperator*{\argmax}{arg\,max}

\AtEndEnvironment{proof}{\qed} 

\newtheorem{assumption}{Assumption}
\newtheorem{observation}{Observation}
\newtheorem{claim}{Claim}

\begin{document}    

\title{A Mechanism for Participatory Budgeting With Funding Constraints and Project Interactions}
\author{Mohak Goyal\thanks{In alphabetical order. Emails: mohakg@stanford.edu, sahasras@stanford.edu, and ashishg@stanford.edu. 
We thank Lodewijk Gelauff and Sanath Kumar Krishnamurty for their insightful comments on this work. }  \and Sahasrajit Sarmasarkar$^{\star}$ \and Ashish Goel}
\authorrunning{Mohak Goyal, Sahasrajit Sarmasarkar, and Ashish Goel}
\institute{Stanford University}
\maketitle
\begin{abstract}

 Participatory budgeting (PB) has been widely adopted and has attracted significant research efforts; however, there is a lack of mechanisms for PB  which elicit project interactions, such as substitution and complementarity, from voters. Also, the outcomes of PB in practice are subject to various minimum/maximum funding constraints on `types' of projects. There is an insufficient understanding of how these funding constraints affect PB's strategic and computational complexities. We propose a novel preference elicitation scheme for PB which allows voters to express how their utilities from projects within `groups' interact. We consider preference aggregation done under minimum and maximum funding constraints on `types' of projects, where a project can have multiple type labels as long as this classification can be defined by a 1-laminar structure (henceforth called 1-laminar funding constraints). Overall, we extend the Knapsack voting model of Goel et al. \cite{goel2019knapsack} in two ways -- enriching the preference elicitation scheme to include project interactions and generalizing the preference aggregation scheme to include 1-laminar funding constraints. We show that the strategyproofness results of Goel et al. \cite{goel2019knapsack} for Knapsack voting continue to hold under 1-laminar funding constraints. Moreover, when the funding constraints cannot be described by a 1-laminar structure, strategyproofness does not hold. Although project interactions often break the strategyproofness, we study a special case of vote profiles where truthful voting is a Nash equilibrium under substitution project interactions. We then turn to the study of the computational complexity of preference aggregation. Social welfare maximization under project interactions is NP-hard. As a workaround for practical instances, we give a fixed parameter tractable (FPT) algorithm for social welfare maximization with respect to the maximum number of projects in a group when the overall budget is specified in a fixed number of bits. We also give an FPT algorithm with respect to the number of distinct votes.

\end{abstract}

\section{Introduction}
Participatory Budgeting (PB)  
is a process through which residents can vote on a city government's use of public funds \cite{cabannes2004participatory,wainwright2003,ganuza2012}. Residents might, for example,
vote on allocating funds between projects like
street repairs or enhancing public safety. PB has been shown to promote
citizen engagement, government transparency, and good governance \cite{wampler2007,johnson2021testing,gatto2021governance}.
Projects in PB have a fixed cost, and there is an overall budget $B$ of funds that the city can spend. 

Several voting methods have been used in PB \cite{benade2021preference}. The most widely used methods in practice are  \emph{K-Approval} \cite{brams2002voting}, in which voters approve up to $K$ projects on the ballot and 
\emph{Approval} \cite{brams1978approval}, in which voters approve any number of projects that they like. These methods are preferred for their simplicity for the voters. 
In \emph{Ranking} \cite{arrow2010handbook}, voters rank all the projects in order of value or value-for-money. 
PB organizers often want voters to understand the budgetary constraints they face. For this, a popular choice is \emph{Knapsack voting} \cite{goel2019knapsack}, in which voters select any number of projects, subject to their costs satisfying the budget constraint $B$. 

However, all these methods ignore utility interactions from different projects.
In fact, most existing work in PB (with the notable exceptions of \cite{jain2020,fairstein2021proportional,rey2020designing}) assume that the utilities of voters are additively separable over different projects. This assumption fails to capture many real-world complexities of voter preferences \cite{aziz2021participatory,rey2023computational}. 
For example, consider the following two projects proposed to enhance public safety in area A.
\begin{example} \label{eg1}
    Project 1: Install more streetlights in area A. \\
    Project 2: Hire additional police officers for area A. 
    
    Some voters may believe that either of these projects is sufficient to solve the problem, but doing both would be excessive. For these voters, the two projects are \emph{substitutes.}  
    However, some other voters may think that both projects are necessary to make area A safe -- and doing only one project would be a waste of money since, in that case, they will continue to avoid the area. For these voters, the two projects are \emph{complements.} 
    Another group of voters may think that additional police officers are not required and only streetlights are required for the area. For these voters, these two projects are \emph{independent} (and they like only Project 1.)
\end{example}

Example \ref{eg1}, with only two projects, illustrates the complexity of eliciting different voters' preferences in PB in real-world scenarios. All the methods we discussed earlier, i.e., K-Approval, Approval, Ranking, and Knapsack voting, fail to capture the different preferences of the voters in Example \ref{eg1}.

There is another natural type of project interaction, which occurs when some projects contradict each other. Consider the following example for road development.

\begin{example} \label{eg2}
    Project 1: Widen the car lane at street X.\\
    Project 2: Build a bike lane on street X. 
    
    Often due to physical constraints, at most one of these projects can be done. In such cases, these projects are \emph{contradictory}. Ideally, the PB ballot \emph{must} inform the voters of such a real-world constraint and restrict the space of their possible votes accordingly.
\end{example}

Previous works \cite{jain2020,fairstein2021proportional} have modelled project interactions in PB via various utility functions. However, their preference elicitation schemes do not enable voters to express their opinions on various project interactions. For example, Jain et al. \cite{jain2020} use Approval voting and assume that the PB organizer knows the project interactions. Fairstein et al. \cite{fairstein2021proportional} allow voters to express their own groups of projects with interacting utilities; however, they only consider substitution project interactions. This paper aims to fill this gap and design an intuitive preference elicitation scheme under which voters can express a broader range of project interactions than in the existing literature.

Another contribution of this paper is the study of the implications of funding constraints on PB's strategic and computational aspects.
From a fairness and equity point of view in budgetary tasks, imposing maximum and/or minimum funding constraints on types of projects is often desirable. Moreover, a project can have multiple `types'. For this, we consider a 1-laminar structure of type labels of projects. As we show in Observation~\ref{obs:tree}, when the funding constraints are defined on this 1-laminar structure of project labelling, these can be represented as a rooted tree -- which, in turn, implies a hierarchical ordering of these project labellings.

Such hierarchically defined funding constraints are natural to study for PB \footnote{See, e.g., \href{https://www.cbo.gov/system/files/2023-03/58890-Discretionary.pdf}{here} and \href{https://www.cbo.gov/system/files/2023-03/58889-Mandatory.pdf}{here} that the US federal discretionary and mandatory fundings are described hierarchically.}. The highest order constraint is the overall budget constraint $B$. The second-order constraint may, for example, be on the funding to different districts in the city. Further down, for each district, we may have constraints on department-level maximum and minimum funding, such as for public safety, infrastructure, and outreach. Within each department, there may be funding constraints on different sub-departments. Each sub-department may be proposing several projects on the ballot. Therefore, the 1-laminar structure of funding constraints closely captures the situations that may arise in real-world PB.
%
We study how these funding constraints affect the strategic considerations\footnote{Note that the funding constraints do not restrict the space of possible votes  -- these constraints are imposed only on the PB outcome. It is conceivable that a PB organizer may impose these constraints on the votes too. However, we believe that this is not a good idea. It makes voting complicated and restricts voters' freedom to express their opinions. For example, in a PB election in Austin, TX, USA, in 2020, certain groups expressed dissatisfaction with the limitations of the budget input tool as to how much funding they could deduct from the police department \cite{austin2020}. The city then conducted a follow-up election removing those constraints from the ballot.} of voters and the computational complexity of preference aggregation algorithms. 

\subsection{Overview of Our Proposed Model}

Same as Jain et al. \cite{jain2020}, we adopt a framework wherein the projects on the ballot are partitioned 
into \emph{groups}\footnote{Disambiguation: see that the `group' and the `type' of a project are different concepts in our model. The groups are designed on the ballot interface for the purpose of preference elicitation regarding projects with potentially interacting utilities. A project is in a single group. Whereas a project `type' is a label imposed onto it for the purpose of specifying funding constraints. A project can have multiple types.} by the PB organizer.
These groups would typically correspond to a theme -- for example, public safety, road development, food support, etc.\footnote{The process of designing the ballot and partitioning the projects into groups is very important -- the fairness and effectiveness of the PB election crucially depend on it. However, it is not the subject of study in this work. 
Jain et al. \cite{jain2021partition} propose doing a preliminary election for aggregating the project partitions, and Baumeister et al. \cite{baumeister2021complexity} study the complexity of manipulating the ballot by its designers.}
We assume that projects in different groups have non-interacting utilities for all voters; however, for projects within a group, there can be one of four possibilities regarding project interactions: the projects are \emph{1) substitutes, 2) complementary, 3) contradictory, or 4) independent (no interaction).} We formally describe the class of preferences expressible in our proposed method in Observation~\ref{obs:class_of_pref} in \S\ref{sec:utility}.
Crucially, in our model, voters need not agree on these project interactions except when the projects are contradictory (in this case, the ballot is designed to reflect the interaction).

Votes in our scheme have three parts. First, voters 
\emph{allocate funds} to different groups of projects subject to the constraint that the sum of these allocations is, at most, the total funds $B$. This part of the vote is inspired by Knapsack voting.
In the second part of the vote, they perform \emph{approval voting} within all groups to which they allocate nonzero funds. 
In the third part, they answer a \emph{complementarity question} for each group to which they allocate nonzero funds, where they give a yes/no answer to whether the projects they approve in the group are complementary \emph{for them.} 
We explain in \S\ref{sec:utility} how these three simple parts of the voting scheme come together and provide an intuitive language for the voters to express project interactions. 


We adopt a natural generalization of the overlap utility function, which was first given by Goel et al. \cite{goel2019knapsack} for Knapsack voting\footnote{Recall that Knapsack voting entails a scheme  where voters approve as many projects they like as long as the cost of their approval set is at most $B$. Such an approval set is referred to as their `preferred' or `ideal' outcome.}. The overlap utility function for Knapsack voting captures, in dollar terms, the agreement between the voter's `ideal outcome' and the PB outcome. Without project interactions, maximizing the overlap utility is equivalent to minimizing the $\ell_1$ distance between the chosen outcome and the votes. Our generalization of the overlap utility (Definition~\ref{def:utility}) captures project interactions; that is, for a bundle of complementary projects, a voter receives utility only if the entire bundle is funded, and for a bundle of substitute projects, the voter's utility saturates at a point that they specify via their vote. 

We take a utilitarian approach and consider preference aggregation for social welfare (SW) maximization, where SW is the sum of the utilities of all voters.
 As is crucial for budgeting tasks, we also incorporate 1-laminar funding constraints to the preference aggregation scheme. We call our mechanism ``Participatory budgeting with project interactions'' (PBPI).

\subsection{Our Results}

We study the incentives of strategic voting in PBPI (\S\ref{sec:IC}). For singleton groups, PBPI is same as Knapsack voting. We show that the strategyproofness\footnote{A voting mechanism is strategyproof \cite{fudenberg1991game} if it is a weakly-dominant strategy for all voters to vote truthfully.} result of Goel et al.\footnote{Goel et al. \cite{goel2019knapsack} showed that Knapsack voting is strategyproof for unit-cost projects under the overlap utility function. Further, their model does not consider project interactions, equivalent to PBPI with singleton groups.} \cite{goel2019knapsack} continues to hold under 1-laminar funding constraints. This result has an important implication: the 1-laminar funding constraints do not make the voting mechanism more complex for a voter. These constraints may, for example, be shared in a separate document from the actual ballot, and the voter may or may not review it to make an `informed' decision. Moreover, this result is tight in the sense that if the the funding constraints do no have a 1-laminar structure, then the mechanism is not strategyproof.

With project interactions, PBPI is often not strategyproof; we study an interesting special case of vote profiles for which truthful voting is a Nash equilibrium \cite{nash1950equilibrium} under substitution project interactions. The arguments used in the proof are subtle and require the construction of a 2-layer algorithm and a careful analysis of voter strategies and potential benefits from various possible strategic deviations.

We then study the computational complexity of preference aggregation in PBPI. Due to the complex project interactions captured in our model, identifying a social welfare maximizing set of projects under budget constraints is NP-hard.
We show that when the project groups have a fixed maximum size, and the budget is specified in a fixed number of bits, the problem is fixed-parameter tractable (FPT). For this, we provide a recursive algorithm (\S\ref{sec:complexity}). Notably, this result holds also with 1-laminar funding constraints. This result implies that for most real-world instances of PB (where not too many projects are expected to have interacting utilities together), computational complexity is not a worry for PB designers, even if the voter turnout is large. We also study the case where the number of distinct votes is small, and give an FPT algorithm for it. While not applicable to general PB elections where all voters submit independent votes, this result is important for cases where a small number of elected representatives engage in a budgetary tasks. Each voter, in this case, may also have a `weight' denoting how many people they represent.

%
\subsection{Related Work}


\textbf{On Project Interactions in PB. }
Jain et al. \cite{jain2020} was the first work to propose a model of PB where projects are divided into groups, and only the projects within a group can have interacting utilities. 
 Our model PBPI differs from theirs in the following ways:  
\begin{itemize} [leftmargin = 0.35cm]
        \item They use approval voting. 
     This choice of the preference elicitation method implies that the PB organizers need to assume the knowledge of project interactions. This further implies that in the eyes of the mechanism, all voters have the same project interactions. Our preference elicitation scheme allows different voters to have different project interactions, which they can express in the vote.  Overall, in the trade-off between expressiveness and simplicity, we adopt a more expressive method, whereas theirs is simpler.
    
    \item Their utility function does not consider project costs -- it depends only on the number of projects funded from a voter's approval set. Our generalization of the overlap utility function accounts for project costs. This is important when projects have vastly different scales, costs, and utilities. 

    \item PBPI models contradictory projects explicitly, unlike their model.
  \item In addition to an overall budget constraint (as considered in Jain et al. \cite{jain2020}), we also consider 1-laminar funding constraints on the PB outcome.
\end{itemize}

Fairstein et al. \cite{fairstein2021proportional} take an egalitarian approach and give a mechanism which gives a \emph{proportional} \cite{peters2020proportionality} outcome while accounting for substitute projects, where voters can express which projects are substitutes.

Rey et al. \cite{rey2020designing} consider perfect complements and perfect substitute project interactions in PB via the solution concept of judgement aggregation \cite{list2009judgement}. 
However, they also have approval votes, and therefore the project interactions need to be assumed to be known to the PB organizers. 

For multi-winner elections (MW) (MW is the same as PB if all projects have equal costs), Izsak \cite{izsak2017working} studies a model with complementarity effects or ``synergies'' and propose a preference elicitation scheme based on the submodular degree (a concept coined by Feige and Izsak \cite{feige2013welfare} to capture the extent to which a function exhibits supermodular behaviour).  
Izsak et al. \cite{izsak2018committee} extend the model to categorize candidates into classes and consider interclass and intraclass synergies; they do not propose a preference elicitation scheme.

\textbf{On Strategic Voting. }
Goel et al. \cite{goel2019knapsack} study the Knapsack voting model with overlap utilities and show that their mechanism is strategyproof if one of the two assumptions holds: 1) partial funding of projects is allowed, 2) projects have unit cost. We adopt a generalization of their overlap utility function, adapted for capturing project interactions (Definition~\ref{def:utility}) and extend their model for capturing project interactions and studying the computational and strategic implications of funding constraints on the outcome. 

Freeman et al. \cite{freeman2021truthful} study a class of ``moving-phantom'' strategyproof mechanisms for PB. They show that the social welfare-maximizing mechanism is the unique Pareto-optimal one in this class of mechanisms.

Yang and Wang \cite{yang2018multiwinner} study strategyproofness in MW elections 
for various types of restrictions on the  outcome, represented by a graph of the alternatives. For example, they consider cases where the outcome has to be a connected subgraph, an independent set, or a subgraph of bounded diameter. Their results do not cover the class of constraints we study, 1-laminar funding constraints.   




\textbf{On Funding Constraints in PB. }
Patel et al. \cite{patel2021group} study a more general model than PB called \emph{group fair knapsack}. They consider minimum and maximum constraints on the total weight of items from a type in the outcome and provide approximation algorithms for several problem variants.



Jain et al. \cite{jain2020participatory} consider a model where each group of projects has a budget limit in addition to the overall PB budget limit. While  social welfare maximization is NP-hard in their setup, they give efficient
algorithms for several special cases.  
Chen et al. \cite{chen2022participatory} consider funding constraints in a setup where the overall budget depends on donations from citizens.

 Constraints on the outcome similar in spirit to our funding constraints have also been studied in MW elections \cite{bredereck2018multiwinner,celis2017multiwinner,yang2018multiwinner,bei2022candidate}. Bredereck et al. \cite{bredereck2018multiwinner} study the computational aspects of identifying an outcome under various scoring rules and diversity constraints defined over a 1-laminar classification of candidates. Bei et al. \cite{bei2022candidate} study the fairness-motivated constraints on types of candidates in a setting where the size of the outcome is not predetermined and can be set to satisfy all fairness-motivated constraints.

Our model of funding constraints (on a 1-laminar classification) is most similar to that of Brederek et al. \cite{bredereck2018multiwinner}. Their focus is on the computational aspects of the model; in addition to that, we also study the strategic aspects of voting. Also, our computational results hold for PB and not just MW elections. 

\begin{figure*}[tb]
\centering
    \includegraphics[scale = 0.41]{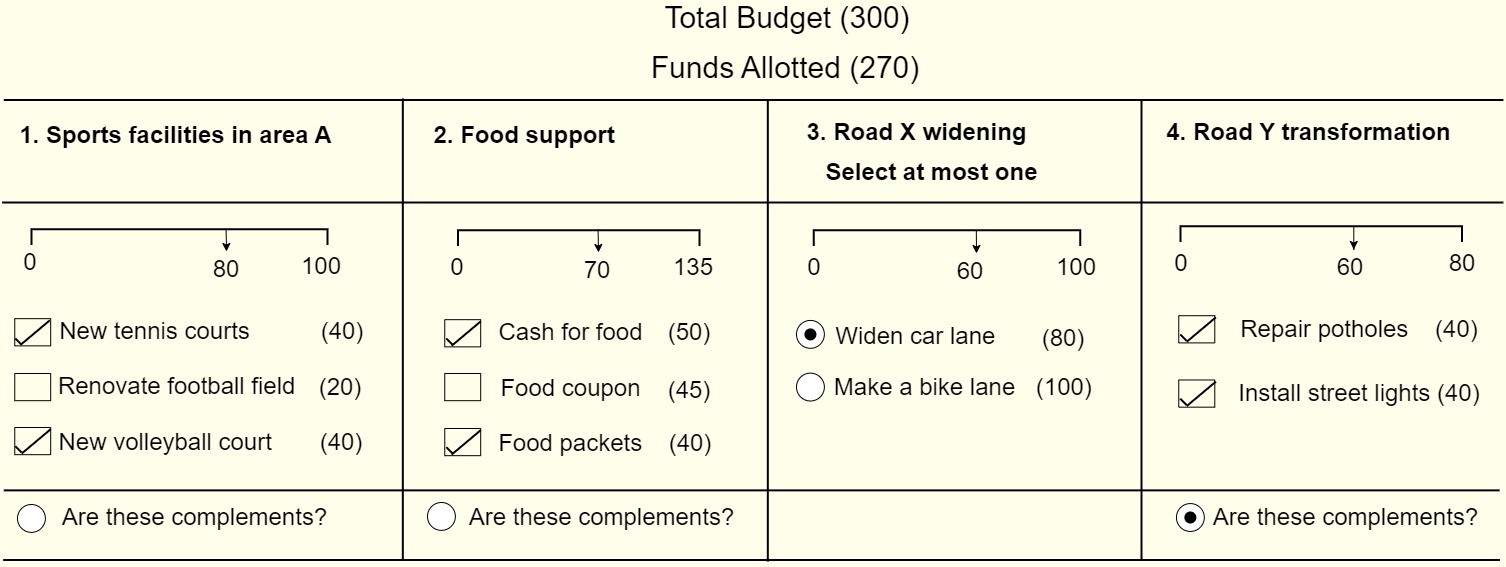}
    \caption{A simple example of a PBPI ballot with a partial vote marked. There are four groups of projects. Groups 1, 2, and 4 have non-contradictory projects. Group 3 is of contradictory projects; a voter can approve at most one project from this group. The total budget is $B = 300.$ The costs of projects are given in parentheses next to its name. The fund allocation $f_g^i$ to each group is represented by a slider. 
    The total funds allocated by the voter are represented by `Funds Allotted' and are constrained to be at most $B.$ 
    Finally, there is a complementarity question, via which the voter can indicate if their approved projects in the group are complementary for them.
    } 
    \label{fig:ballot}
\end{figure*}
\section{Model}
\label{sec:PBPI}
In this section, we describe the design of the ballot, our novel preference elicitation scheme, the utility function, the preference aggregation scheme, and conclude the section with a discussion of the merits and limitations of the model. Figure~\ref{fig:ballot} is an example of a simple ballot in our mechanism. 

\subsection{Ballot Design and Preference Elicitation Scheme} \label{sec:ballot}
There are $n$ voters $\{1,2,\ldots,n\}$ denoted by $[n]$ and $m$ projects given by $P = \{p_1, p_2, \ldots, p_m\}.$ Project $p_j$ has a fixed positive rational cost $c_j$. $ B$ denotes the total budget of funds.
 An \emph{outcome} of an instance of PB is a bundle of projects $Q$ that satisfies the budget constraint $ \sum_{j \in Q} c_j \leq B.$
 The PB organizer partitions the set of projects $P$ into $r$ \emph{groups} --  this partition is $Z = \{z_1, z_2, \ldots, z_r\}$ such that $\cup_{j \in [r]} z_j = P$ and $z_j \cap z_k = \emptyset$ for all $j \neq k.$ Groups of projects on the ballot can be of one of two forms: \emph{contradictory} and \emph{non-contradictory.}

In a group of contradictory projects, there is a real-world constraint that at most one of the projects can be implemented\footnote{We can generalize this constraint  such that at most $k$ of these projects can be implemented for any $k > 0.$ $k$ can be different for different groups of projects. All our results will continue to hold. We use $k=1$ for clarity of exposition.} (recall Example \ref{eg2} on road widening). PBPI imposes this constraint on all voters.

%
%

Voter $i'$s vote $v^i$ has three components: $(f^i, s^i, t^i).$
\begin{itemize}
    \item The \emph{fund allocation} for group $g,$   $f^i_g \geq 0$ is the amount of funds that voter $i$ allocates to $g$. $f^i$ must satisfy the budget constraint $\sum_{g \in [r]} f^i_g \leq B.$
    \item The \emph{approval set} in group $g,$ $s_g^i,$ is the subset of $z_g$ that voter $i$ approve.
\item Finally, $t^i_g$ is the binary answer to the \emph{complementarity question} in group $g$, such that $t^i_g = 1$ if voter $i$ considers $s_g^i$ to be complementary, and $t^i_g = 0$ if voter $i$ considers $s_g^i$ to be substitutes or independent.
\end{itemize}
 The complementarity question is not asked for groups of contradictory projects (see, for example, Group 3 in Figure~\ref{fig:ballot}) where at most one project can be done -- here $t^i_g$ is fixed to $0$ for notational convenience.
We refer to the set of all votes $\{v^i | i \in [n]\}$ as the \emph{vote profile} $V.$ We further assume that   the project costs $c_j,$ total budget $B,$ and fund allocations $f_g^i$  $\forall i \in [n], g \in [r]$ are positive integers. This is reasonable since the parameters $C$ and $B$ can be rescaled to be integers as long as these are rational without changing the outcome of PB. Typically in real-world PB, costs and budgets are not specified at high precision.
%
\subsection{Utility Function and the Class of Expressible Preferences} \label{sec:utility}
We adopt a generalization of the overlap utility \cite{goel2019knapsack}. In this subsection, we describe how our preference elicitation method enables voters to express their project interactions and also define the class of expressible preferences in PBPI.
\begin{definition}[Utility] \label{def:utility}
    The utility of voter $i$ with vote $v^i = (f^i, s^i, t^i)$ from a bundle of projects $Q$ is:
    \begin{equation*}
 u_i(Q) = \sum_{g=1}^r 
 \left(t_g^i \cdot \mathbb{I}(s_g^i \subseteq Q) \cdot \min (f_g^i, \sum_{j \in \{Q \cap s_g^i\}} c_j )  + (1- t_g^i)\cdot\min (f_g^i, \sum_{j \in \{Q \cap s_g^i\}} c_j )\right)  \label{eq:utility}
\end{equation*}
Here $\mathbb{I} (\cdot)$ denotes the indicator function.
\end{definition}

\begin{definition}[Social Welfare] \label{def:social-utility}
The sum of utilities of all $n$ voters is the social welfare $U(Q)$, i.e., $U(Q) = \sum_{i \in [n]} u_i(Q).$ We use the terms social welfare and social utility interchangeably.
\end{definition}

A voter's utility from group $g$ cannot exceed $\sum_{j \in  s_g^i} c_j$. Therefore, we may restrict the amount of funds they allocate to a group, $f_g^i,$ at be most $\sum_{j \in s_g^i} c_j.$ However, such a condition is not required for the technical analysis, and we omit it. All our results continue to hold if we impose such a constraint on $f_g^i.$

The first term in the utility function corresponds to the groups of \emph{complementary projects} for the voter (i.e., $t^i_g=1$). In this case, the voter receives utility equal to $\min (f_g^i, ~\sum_{j \in \{Q \cap s_g^i\}} c_j )$ if all the projects they approve in this group are funded. Otherwise, they get zero utility from the group.  Group 4 in Figure~\ref{fig:ballot} is an example of this case for the marked vote. The voter believes that repairing the potholes and installing the streetlights are both required to make road Y useable -- doing any less will be a waste of money.


The second term captures the utility structure of \emph{substitution with satiation.} This function is inspired by the Leontief-free utility function of Garg et al. \cite{garg2018substitution}. 
In our model, interestingly, this term can capture the other three cases of project interactions - contradictory, substitute, and  independent groups of projects. 
 
 In \emph{contradictory groups}, voters are allowed to approve at most one project. If voter $i$ approves project $j^* \in z_g$, their utility from group $g$ is $\min(f^i_g, c_{j^*})$ if $j^* \in Q$  and $0$ otherwise. See, e.g., Group 3 in Figure~\ref{fig:ballot}. 
 
 For non-contradictory groups, if a voter believes that the projects they approve are \emph{substitutes}, they allocate as many funds $f^i_g$ at which their utility is capped (i.e., their point of satiation). Their utility from the projects in $s_g^i$ adds up linearly up to the point where it saturates. Overall, it is $\min (f_g^i, \sum_{j \in \{Q \cap s_g^i\}} c_j ).$ 
 An example is Group 2 in Figure~\ref{fig:ballot}. The voter in the example thinks that distributing food packets and giving out cash for food are both useful projects, and the utility they derive from these projects is up to $70$. Note that this is less than the sum of the costs of the projects they approve in the group.
 
 However, if the projects are \emph{independent} for them, voters can express this by allocating funds $f_g^i$ that cover the cost of all the projects they approve in the group. 
 An example is Group 1 in Figure~\ref{fig:ballot} for the marked vote. Here, the voter thinks that new tennis courts and a new volleyball court are both useful, and the utilities are independent. They allocate funds $f^i_g$ equal to the costs of the two projects so that their utility function from this group is not saturated at $f^i_g$.

 See that PBPI is a generalization of the knapsack voting framework of Goel et al. \cite{goel2019knapsack}. If all projects are \emph{independent,} then a voter can allocate $f^i_g$ equal to the sum of costs of the projects they approve in the group and set $t_g^i = 0$. PBPI can, in fact, elicit a much richer class of preferences from individual voters, as we describe in the following.

 \textbf{Class of Expressible Preferences}

The most expressive PB scheme is one where voters report their utility for each possible subset of projects. This requires $2^m$ entries from each voter and is infeasible except for very small ballots.
On the other hand, Approval and K-Approval voting methods focus on simplicity and disregard the projects' costs and project interactions. 
Our mechanism PBPI takes the middle ground between the two extremes in the trade-off between expressiveness and simplicity.
 We summarize the class of preferences expressible in PBPI below.
 \begin{observation}\label{obs:class_of_pref}
 For a group of contradictory projects, a voter can express their:
\begin{itemize}
\item \textbf{A. Top choice}, and the utility they get from it, which can be up to its cost. 
\end{itemize}
For non-contradictory groups, a voter can express the following preferences:
    \begin{itemize}
    \item \textbf{B1. Perfect complements:} The voter views the set of projects $s^i_g$ as one unit, and their utility from this unit is $f_g^i$ (or the cost of the unit, whichever is lower) if this unit is implemented entirely and $0$ otherwise.
    \item \textbf{B2. Perfect substitutes with satiation:}  The utility from each project is equal to its respective cost, but their total utility from the group saturates at $f_g^i$. 
    \item \textbf{B3. Independent:} The voter views all projects in $s_g^i$ as independent, and the utility from each equals its cost (as per the standard Knapsack model \cite{goel2019knapsack}).
    \end{itemize}

    The cases B2 and B3 are differentiated in the vote by the satiation level expressed in $f_g^i.$ For $t_g^i = 0$, when $f_g^i$ is at least equal to the sum of the costs of all the projects approved by the voter, the group is of independent projects for them, and substitute projects otherwise.
 \end{observation}

\subsection{Preference Aggregation and Funding Constraints} \label{sec:gf}
A line of work on PB \cite{jain2020participatory,patel2021group} has studied the problem of imposing minimum and/or maximum \emph{funding constraints}  on types of projects to ensure the diversity of the outcome. This is sometimes done with the idea that a particular type of projects may be more beneficial to certain demographics in society. Since fairness-motivated interventions based on demographics are often hard to implement in PB \cite{gelauff2020advertising} or are disallowed by law \cite{sowell2004affirmative}, imposing funding constraints on types of projects is a reasonable proxy to obtain a diverse and equitable outcome. 

Let $L$ be a set of labels. Denote the set of projects with the label $l \in L$ by $P_l.$ We refer to the set of projects $P_l$ as `type' $l$. Each project can have any number of labels from set $L$. These labels could, for example, indicate the project's geographic location or theme (e.g. infrastructure, education, etc.). However, this general system of project labels is difficult to study formally for computational complexity and strategic voting. Therefore we make the following assumptions. 


\begin{assumption}[1-laminar labelling] \label{ass:1lam}
The labelling is 1-laminar\footnote{Bredereck et al. \cite{bredereck2018multiwinner} considered the same labelling structure  while studying computational complexity in MW.}. That is, for any two distinct labels $x$ and $y$, the sets $P_x$ and $P_y$ satisfy either a) $P_x \cap P_y = \emptyset$ or  b) $P_x \subset P_y$ or c) $P_y \subset P_x$.
\end{assumption}

\begin{observation} \label{obs:tree}
    Any 1-laminar labelling can be represented as a rooted tree, denoted by $\mathcal{T}_L.$ There is a node corresponding to each label. The root node corresponds to a default label (which we construct for notational convenience) that applies to all projects. Nodes of all labels $y$ such that $P_y \subset P_x$ are children of the node representing label $x$ if there is no $z$ satisfying  $P_y \subset P_z \subset P_x.$ 
\end{observation}
\begin{figure}[tb]
    \centering
\includegraphics[scale = 0.35]{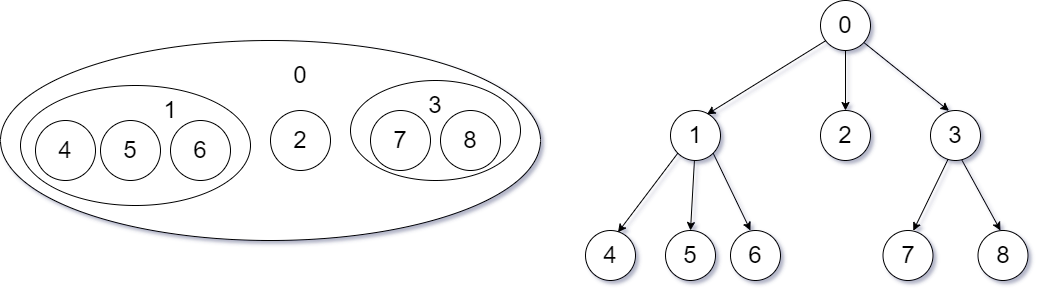}
    \caption{Each node corresponds to a label or `type' of projects. The set of projects with label 1 is the union of the set of projects with labels 4, 5, and 6.  All projects have the default label 0. The 1-laminar structure of the labels enables the representation of the label relations as a rooted tree as described in Observation~\ref{obs:tree}.}
    \label{fig:laminar}
\end{figure}

\begin{assumption} \label{ass:group_label}
    All projects in a group on the ballot have the same labels.
\end{assumption}
This assumption implies that we can introduce a new level at the bottom of the rooted tree of the labels wherein  the `group' represents a new type label and forms the leaves of the tree. 

\begin{definition}[1-laminar funding constraints] \label{def:gf}
     A valid outcome  $Q$  of the PB election must satisfy the following constraints:
     \begin{align}
     B^{\min}_l \leq \sum_{j \in Q \cap P_l} c_j \leq B^{\max}_l 
 \ \ \ \ \ \ \ \ \ \ \ \text{for all~} l \in L. \label{eq:gf}
\end{align}
\end{definition}



Note that the overall budget constraint can be seen as being associated with the default label, which applies to all projects on the ballot and corresponds to the root of the tree representation of the 1-laminar labelling.
Later in the paper, we will study if including the 1-laminar funding constraints imposes additional challenges for the computational complexity of preference aggregation or have consequences regarding strategic voting. We have positive results on both fronts - Theorems~\ref{thm:nash_equilibrium},~\ref{thm:fpt}, and~\ref{thm:fpt-n}.


We consider mechanisms for PB that produce \emph{non-fractional} and \emph{deterministic} outcomes. 
Further, we take a utilitarian approach and consider PB outcomes that \emph{maximize the social welfare (SW)}. 
In the special case where 1-laminar funding constraints are not imposed on the outcome, we denote the problem of identifying an SW-maximizing bundle of projects under budget constraints by SWM-PB;
with 1-laminar funding constraints, we call this problem FC-SWM-PB.

\subsection{Limitations of the Model} \label{sec:discussion}
%




Same as the previous works in this line \cite{jain2020,fairstein2021proportional}, PBPI does not capture the case where both substitutes and complementary projects are present in a group. For example, in a group of projects $\{p_1,p_2,p_3\}$, it is possible that a voter finds $p_1$ and $p_2$ to be complementary, but $p_3$ to be a substitute for the bundle $\{p_1,p_2\}.$ PBPI does not enable voters to express such project interactions. 

PBPI also does not model \emph{ranked preferences over substitute projects}, which could have been elicited via a ranking-based voting scheme within groups instead of approval in PBPI.

Further, our utility function does not model \emph{soft complements}, i.e., per our utility function, a voter cannot get partial utility if only a part of their approval set in a group of complementary projects is funded. 

For a group $\{p_1,p_2\},$ voters in PBPI cannot express preferences of the following type: ``$p_1$ and $p_2$ are independent for me, but my utility from each is only half as much as their respective costs.'' This is because our utility function cannot distinguish between this case and the case where $p_1$ and $p_2$ are substitutes for the voter. A voter could express this type of preference if $p_1$ and $p_2$ were in singleton groups on the ballot.


We can mathematically model the above-stated and even more complex project interactions with cost-based utility functions. Our positives result on computational complexity in Theorems~\ref{thm:fpt} and~\ref{thm:fpt-n} (FPT with respect to the maximum group size) will hold for any utility function as long as all project interactions are within their own groups. However, we do not adopt a more complex utility function since it would require an equally expressive preference elicitation method which may not be intuitive for the voters in the real world.

Having discussed the model of PBPI and its merits and limitations, we now move on to investigating the incentives for strategic voting in PBPI.

\section{Incentives for Strategic Voting in PBPI} \label{sec:IC}

Goel et al. \cite{goel2019knapsack} showed that Knapsack voting is strategyproof under the overlap utility model if one of the two assumptions holds: 1) partial funding of projects is allowed, 2) projects have unit cost. 
Technically, these assumptions provide the same leverage, ensuring that a greedy algorithm is optimal for preference aggregation. In this section\footnote{We do not need this assumption in the section on computational complexity
.}, we work with the assumption that projects are unit-cost\footnote{It is easy to show that under overlap utility, the unit-cost assumption is crucial for strategyproofness to hold in Knapsack voting and PBPI. We sketch an example here. Consider a case with four projects whose costs are $c_1 = 1$, $c_2 = 2$, $c_3 = 2$, and $c_4 = 3$. Budget $B$ is $4$ units. There are no project interactions. A voter who only likes project $4$ is incentivised to vote for project $1$ too when all projects have equal approvals otherwise.}.

\begin{assumption}[Unit-Cost Projects] \label{ass:unitcost}
     All projects have cost $c_j =1$  $\forall j \in [m].$
\end{assumption}

Therefore, our results in this section are for the MW election setting. In our first result, we consider the special case of PBPI without project interaction, i.e., the same as the Knapsack voting setup. In this case, we show that even with 1-laminar funding constraints on the outcome, Knapsack voting is strategyproof, i.e., voters are incentivised to vote truthfully, disregarding the 1-laminar funding constraints\footnote{Recall that the votes need not satisfy the 1-laminar funding constraints, but only the overall budget constraint.}.

\begin{theorem} \label{propostion:strategyproof}
With singleton groups, unit-cost projects (Assumption~\ref{ass:unitcost}), and the 1-laminar labelling (Assumption~\ref{ass:1lam}) for the funding constraints (Definition~\ref{def:gf}), PBPI is strategyproof.
\end{theorem}

 The proof uses Theorem~\ref{thm:nash_equilibrium} (given later) with some discussion which we give in Appendix~\ref{sec:knapsack_lam}. Note that the result also extends to the case of fractional knapsack considered by Goel et al. \cite{goel2019knapsack}, since each project can then be considered a collection of unit-cost projects.

 The 1-laminar labelling (Assumption~\ref{ass:1lam}) is a fairly general setup for defining the funding constraints in practice. Also, as it turns out, technically, Assumption~\ref{ass:1lam} is necessary for Theorem~\ref{propostion:strategyproof} to hold. Towards this, we have the following result, proof of which is given in Appendix~\ref{sec:nonlam}.

\begin{theorem}\label{obs:non-laminar}
    PBPI with singleton groups and unit-cost projects is not strategyproof if the labelling  for the funding constraints does not satisfy the 1-laminar structure (Assumption~\ref{ass:1lam}).
\end{theorem}
We now study the effect of considering project interactions on the incentives to vote truthfully.
Unfortunately, as is seen often in social choice theory \cite{gibbard1973manipulation,satterthwaite1975strategy}, voters in PBPI may be incentivised to deviate from their truthful votes to attain a better outcome for themselves when project interactions are considered.

\begin{observation} \label{obs:not_sp}
    PBPI is not strategyproof even with Assumption~\ref{ass:unitcost}. 
\end{observation}
\begin{proof}[Proof with complementarity project interaction.]
    Consider the following example. 
   
        There are $n=3$ voters and $m=6$ projects, each with cost $c_j = 1.$ The budget is $B=3$. There are no 1-laminar funding constraints. There are two  groups of projects $z_1 = \{p_1,p_2,p_3\},$ and $z_2 = \{p_4,p_5,p_6\}.$ In the truthful votes, all projects are independent for all three voters. Voters $1$ and $2$ approve projects $p_1,p_2,$ and $p_4$ and set funds $f^1_1  = f^2_1 = 2$ and $f^1_2 = f^2_2 = 1.$ Whereas voter $3$ approves projects $p_3,p_5,$ and $p_6,$ and sets $f^3_1  = 1$ and $f^3_2 = 2.$ In this case, the outcome of SWM-PB is $\{p_1,p_2,p_4\}.$

   Voter $3$ has the incentive to modify their vote and set $f^3_1 =3$, approve all of $z_1,$ and report that these projects are complements, i.e., $t^3_1 = 1$. Now, $\{p_1, p_2, p_3\}$ is the outcome of SWM-PB -- it increases the utility of voter $3$ by 1 unit.  
\end{proof}

The example above shows how the complementarity project interaction can potentially lead to profitable strategic deviations for voters. However, this type of project interaction is not required to render PBPI not strategyproof.

\begin{observation} \label{obs:not_sp_not_complementary}
    PBPI is not strategyproof under Assumption~\ref{ass:unitcost} when there are only substitutes allowed. 
\end{observation}

Despite the negative result of Observation~\ref{obs:not_sp}, it is interesting to study special classes of vote profiles where voters do not have incentives to deviate from truthful voting even with project interactions.
We now describe such a class of vote profiles.

\subsection{Case with a Strict Total Order in Subgroups of Substitutes} 
Here we give a class of vote profiles for which truthful voting is a Nash equilibrium.
\emph{Subgroups} of a group are defined by an arbitrary partition of the projects in the group. 
In this class of vote profiles, there is no complementarity protect interaction, voters agree on the type of project interaction in a group, and for substitutes, it entails  \emph{some} partition of groups into subgroups such that there is a strict total order\footnote{A strict total order on a set $S$ is a relation on $S$ that is irreflexive, anti-symmetric, transitive, and every pair of elements of $S$ is comparable.} of projects within each subgroup. 
We further need that all voters who approve project $p$ also approve all projects of a higher order than $p$. Formally:

%

%

\begin{definition}[Special Vote Profile]
\label{def:vote_profiles} 
For each non-contradictory group $g$, one of the following holds:

    1. (Independents). All voters $i \in [n]$ consider the group to be of independent projects and set $f^i_g = |s_g^i|.$
    
    2. (Substitutes). There exists a partition of $z_g$ into ``sub-groups'' and there is a  strict total order in each sub-group such that every voter who approves a project $p \in z_g$ also approves all projects of a higher order than $p$. No voter approves projects from multiple subgroups in a group.
\end{definition}
%
%
Example~\ref{ex:vote_profile} 
illustrates the vote profiles in Definition~\ref{def:vote_profiles}.

Observe that we do not make any assumption on fund allocations $f^i_g$ for substitute projects. Also, this structure does not assume which groups any voter funds. However, it does assume that voters who allocate funds to a group agree on the type of project interaction in any group.
Also, note that having a strict total order on some partition of a group into subgroups is a strictly weaker condition than having a strict total order over entire groups of projects.  
As a clarification, while studying the vote profile of Definition~\ref{def:vote_profiles}, we do not make any assumption on the space of expressible votes, other than that there are no complementarity project interactions considered.

The class of vote profiles in Definition~\ref{def:vote_profiles} is important to study since voters who are residents of an area may agree on the relative usefulness of projects in a group but may not agree on how many projects must be done from any group.
%


%


Having explained the case of vote profiles of Definition~\ref{def:vote_profiles}, we now discuss our results on it. First, when the vote profile $V$ is per Definition~\ref{def:vote_profiles}, a greedy aggregation algorithm (Algortihm~\ref{alg:greedy-gf}) produces a social welfare-maximizing outcome under an overall budget constraint and 1-laminar funding constraints. This is not particularly surprising since in contradictory and independent groups, utilities of individual projects are additively separable, and in substitute groups, the strict total order in subgroups provides an optimal ordering to fund projects. The 1-laminar funding constraints dictate which projects are eligible for funding at any step of the algorithm. 

The surprising result is that the same thing also holds when one voter's vote doesn't follow Definition~\ref{def:vote_profiles} (Lemma~\ref{lem:greedy}). This one vote could disagree with others on which groups are independents and which are substitutes, could approve projects across subgroups, or violate the strict total order assumption within subgroups. The fact that the greedy algorithm still produces a social welfare-maximizing outcome when one vote does not follow Definition~\ref{def:vote_profiles} has important implications for the study of strategic voting in our model.

We can leverage Lemma~\ref{lem:greedy} to show that when the truthful votes (underlying, not observed)  are per Definition~\ref{def:vote_profiles}, truthful voting is a  Nash equilibrium\footnote{We do not obtain strategyproofness here; if multiple votes violate  Definition~\ref{def:vote_profiles}, then there can arise opportunities for profitable strategic deviations for voters.} (Theorem~\ref{thm:nash_equilibrium}). This result has important practical consequences, which we discuss later in the section after giving the results formally.
\subsubsection{Greedy Algorithm} We first describe the Greedy Algorithm~\ref{alg:greedy-gf}, which we use in Lemma~\ref{lem:greedy} and Theorem~\ref{thm:nash_equilibrium}. First, without the 1-laminar funding constraints, a greedy algorithm is simple; it will construct the outcome $Q$ in $B$ steps, adding a project in each step which brings the best improvement in social welfare.
With 1-laminar funding constraints, our Greedy Algorithm~\ref{alg:greedy-gf} is run in two passes. 
In the first pass, the algorithm fulfils the minimum funding constraints\footnote{We assume that the funding constraints $B^{\min}$ and $B^{\max}$ are such that at least one valid outcome of PB exists.} given by $B^{\min}$. 
Recall from Observation~\ref{obs:tree} that the 1-laminar labelling can be represented as a rooted tree denoted by $\mathcal{T}_L$, and the levels of the tree form a hierarchy of the labellings. 
Within this pass, the algorithm proceeds in the reverse hierarchy order for the type labelling. 

\begin{algorithm}[tb]
\SetAlgoLined
 \SetAlgorithmName{Greedy Algorithm}{}

Denote the set of outcomes which satisfy the 1-laminar `maximum' funding constraints by $\mathbb{Q}$\\
Initialize $Q \leftarrow \emptyset$\\
1. First Pass \Comment{Satisfy minimum funding constraints}\\

\For{Traverse $\mathcal{T}_L$ in reverse order\footnotemark, index node by $l$}{
\While{$|Q \cap P_l| < B^{\min}_l$}{

$Q \leftarrow Q \cup \Big\{ \argmax\limits_{\rho \in \{ j | j \in P_l\setminus Q,~j\cup Q \in  \mathbb{Q}\}}~~U(Q\cup \{\rho\}) \Big\}$

}
}
2. Second Pass \Comment{Complete $Q$ while respecting funding constraints}\\
\While{$|Q| < B$}{
$Q \leftarrow Q \cup \Big\{   \argmax\limits_{\rho \in \{j | j \notin Q,  \  j \cup Q \in \mathbb{Q} \}}~~U(Q\cup \{\rho\})     \Big\}$
}
    \caption{Input: $(Z,B,V,L,P_l~\forall~l \in L, B^{\max}, B^{\min})$, Output: $Q$}
\label{alg:greedy-gf}
\end{algorithm}
\footnotetext{A traversal of a rooted tree in the reverse order is defined as the traversal which starts with the leaves in arbitrary order, followed by the deletion of all the leaves, and repeating the process till the entire graph is traversed. }
In the second pass, the algorithm produces the outcome $Q$ while respecting all the maximum funding constraints while adding one project to the outcome at a time. The algorithm pseudocode is given formally in Greedy Algorithm~\ref{alg:greedy-gf}. We now state and prove the main results of this section.

In this subsection, we assume that there are no complementarity project interactions, and in fact, the complementarity question is disabled and we have $t^i_g$ for all voters $i \in [n]$ and all groups $g \in [r].$

\begin{lemma}  \label{lem:greedy}
Under Assumptions~\ref{ass:1lam},~\ref{ass:group_label}, and~\ref{ass:unitcost}, and when at most one vote deviates from Definition~\ref{def:vote_profiles}, the outcome of the Greedy Algorithm~\ref{alg:greedy-gf} is a solution to FC-SWM-PB.
\end{lemma}

\begin{proof}   

If all projects have additively separable social welfare, Greedy Algorithm~\ref{alg:greedy-gf} is trivially optimal for maximizing social welfare.

As a warm-up, consider the case where all votes are per Definition~\ref{def:vote_profiles}. Three varieties of project groups under Definition~\ref{def:vote_profiles} exist -- contradictory, independent, and substitutes. For contradictory and independent groups, it is easy to observe that the social welfare of all projects is additively separable. 

For groups of substitutes, we define the following notation for the subgroups and the strict total order therein. There is a partition $\Lambda_g = \{\lambda_{g,k}\}_{k \in [q_g]}$ of $z_g$ such that each partition is a subgroup and there is a strict total order in each subgroup. Since each voter approves projects from at most one subgroup $\lambda_{g,k}$ per Definition~\ref{def:vote_profiles}, the social welfare of these subgroups is additively separable. Finally, observe that in any subgroup $\lambda_{g,k}$, from any outcome, if we replace a lower-ranked project $p_l \in \lambda_{g,k}$ with a higher-ranked project $p_h \in \lambda_{g,k},$ then the social welfare of the outcome cannot decrease since under Definition~\ref{def:vote_profiles}, all voters who get utility from $p_l$ also get equal utility from $p_h.$  Therefore, greedily selecting projects in order of rank (given by the strict total order in the subgroup) is optimal for maximizing social welfare from subgroup $\lambda_{g,k}$ for any amount of funding to the subgroup. This proves that greedily selecting the outcome of PB is optimal for SWM-PB in this case, where everyone votes per Definition~\ref{def:vote_profiles}.

We now consider the case where one vote does not satisfy Definition~\ref{def:vote_profiles}. Let this be voter $i$ and let their vote be $\{f^i_g, s_g^i\}_{g \in [r]}.$\footnote{Recall that we assume that there are no complementarity project interactions in this case and there is no $t_g^i$.} Denote the vote profile of the other voters by $V_{-i}.$ Observe that for contradictory groups, even under a deviation from Definition~\ref{def:vote_profiles}, the social welfare of all projects are additively separable. 

Consider non-contradictory groups. Let $Q_{\{g\}}(b)$ denote the size-$b$ \emph{social welfare-maximizing} subset of $z_g.$ While the social welfare-maximizing subset is not always unique, we define $Q_{\{g\}}(b)$ to be unique and decided as per the given deterministic tie-breaking rule. 


Let $\mathcal{P}_{\{g\}}^{\text{in}}(b)$ denote the set $Q_{\{g\}}(b)\setminus Q_{\{g\}}(b-1)$   and   $\mathcal{P}_{\{g\}}^{\text{out}}(b)$ denote the set $Q_{\{g\}}(b-1) \setminus Q_{\{g\}}(b)$.  We first make the following claim.
\begin{claim}  \label{claim:greedy}
If $\mathcal{P}_{\{g\}}^{\text{out}}(b)$ is an empty set for all groups $g \in [r]$ and all levels of funding $b \in [|z_g|],$ then the Greedy Algorithm~\ref{alg:greedy-gf} is optimal to solve FC-SWM-PB. 
\end{claim}
\textit{Proof of Claim~\ref{claim:greedy}:}

    Recall the rooted tree representation of the 1-laminar labelling given by $\mathcal{T}_L$. Recall that the minimum and maximum funding constraints on project type $l \in L$ are given by $B_l^{\min}$ and $B_l^{\max}$. However, since projects can have multiple type labels, the implied constraints can be more stringent depending on the structure of tree $\mathcal{T}_L$. Let $\tilde{B}_l^{\min}$ and $\tilde{B}_l^{\max}$ denote the minimum and maximum possible fund allocations to projects of type label $l$ in any outcome which satisfies the entire set of 1-laminar funding constraints.
    
    We can give a recursive computation for $\tilde{B}_l^{\min}$ and $\tilde{B}_l^{\max}$. For any type label $l$ which is a leaf of $\mathcal{T}_L$, $\tilde{B}_l^{\min}=B_l^{\min}$ and $\tilde{B}_l^{\max}=B_l^{\max}$. For any label $l$ which is not a leaf of $\mathcal{T}_L$, we have $\tilde{B}_l^{\min}=\max(\sum_{l' \in c_l}\tilde{B}_{l'}^{\min},{B}_l^{\min})$ and $\tilde{B}_l^{\max}=\min(\sum_{l' \in c_l}\tilde{B}_{l'}^{\max},{B}_l^{\max})$ where $c_l$ denotes the set of child nodes of node $l$. This follows from the 1-laminar structure. We assume that $B^{\min}$ and $B^{\max}$ are such that a feasible solution exists.

     Observe in the Greedy Algorithm \ref{alg:greedy-gf} that the total funds allotted to projects of type label $l$ is always between $\tilde{B}_l^{\min}$ and $\tilde{B}_l^{\max}$ by construction, and thus our algorithm returns a feasible allocation.

Let $\mathcal{P}_{G}^{\text{out}}(b)$ denote $Q_G(b-1) \setminus Q_G(b)$ for any set of groups $G.$
 $\mathcal{P}_{\{g\}}^{\text{out}}(b)$ being an empty set for all groups $g\in [r]$ and all $b \in [|z_g|]$ implies that  $\mathcal{P}_{G}^{\text{out}}(b)$ is also an empty set for any set of groups $G$ and any $b \in [\sum_{g \in G} |z_g|]$ due to the additive separability of utilities across groups.

Consider the first pass of the Greedy Algorithm~\ref{alg:greedy-gf}, and consider the stage when node $l$ is traversed. At this stage, the sub-tree rooted at node $l$ (except node $l$) satisfies all the minimum funding constraints. Further, at this stage, the amount of funds allocated to projects with label $m,$ where node $m$
is a child of node $l,$ are $\tilde{B}_m^{\min}$. Recall that per Assumption~\ref{ass:group_label}, all projects in a group have the same labels. Let $G_l$ denote the set of groups with label $l.$ Also, let $Q^{\min}_{G_l}(b)$ denote the size-$b$ social welfare maximizing project selection from $G_l$ satisfying all the minimum funding constraints of the sub-tree rooted at node $l$ (except node $l$). Similarly, define $\mathcal{P}^{\text{out},\min}_{G_l}(b)=Q^{\min}_{G_l}(b-1)\setminus Q^{\min}_{G_l}(b)$.

Observe that $Q^{\min}_{G_l}(b)$ is defined only for $b\geq \sum_{l' \in c_l}\tilde{B}_{l'}^{\min}$ and observe that $\mathcal{P}_{\{g\}}^{\text{out}}(b)$ being an empty set for all $g\in [r]$ and all $b \in [|z_g|]$ ensures that $\mathcal{P}^{\min}_{G_l}(b)$ is also an empty set. 
%
The proof now follows from induction over the levels of the rooted tree $\mathcal{T}_L.$
After traversing a leaf node $l$, the set of projects selected is $Q^{\min}_{G_l}(B_l^{\min}).$ 
Now, inducting over the levels of the tree, after node $l$ is traversed, the number of projects selected with type label $l$ is $\tilde{B}_{l}^{\min}$ and is $Q^{\min}_{G_l}(\tilde{B}_{l}^{\min})$.
Therefore, after the first pass of the Greedy Algorithm~\ref{alg:greedy-gf}, the outcome is social-welfare maximizing for size $\tilde{B}_{0}^{\min}$ (recall that the label of the root node of tree $\mathcal{T}_L$ is $0$).



In the second pass, at every stage, observe that the amount of funds allocated to projects of label $l$ is always bounded by $\tilde{B}_l^{\max}.$ 
Hence the 1-laminar funding constraints are always satisfied. By a similar argument as for the first pass, we obtain that a greedy selection of projects in the second pass is also optimal when $\mathcal{P}^{\text{out}}_{G_l}(b)$ is empty for all $b$ and all sets of groups of projects $G_l$.


This completes the proof of Claim~\ref{claim:greedy}.
    Now we go back to proving Lemma~\ref{lem:greedy} using Claim~\ref{claim:greedy}.


%
Claim~\ref{claim:greedy} implies that if the Greedy Algorithm~\ref{alg:greedy-gf} is \emph{not} optimal to solve FC-SWM-PB, then there exists some group $g \in [r]$ and an amount of funding $b \in [|z_g|]$ for which  $\mathcal{P}_{\{g\}}^{\text{out}}(b)$ is not empty. Let $g^*$ be such a group and let $b^* \in [|z_g|]$ be the \emph{smallest} funding amount for which  $\mathcal{P}_{\{g^*\}}^{\text{out}}(b^*)$ is not empty.

Consider the following exhaustive set of cases:\\
\textbf{Case 1.} $ f^i_{g^*} = 0.$ 

In this case, $Q_{\{g^*\}}(\cdot)$ is constructed as if the entire vote profile follows Definition~\ref{def:vote_profiles}, and from the arguments above, it is constructed greedily and therefore  $\mathcal{P}_{\{g^*\}}^{\text{out}}(b^*)$ is empty.\\
   \textbf{Case 2.} There exists a project $p_{\text{in}} \in \mathcal{P}_{\{g^*\}}^{\text{in}}(b^*)$ and a project $p_{\text{out}} \in \mathcal{P}_{\{g^*\}}^{\text{out}}(b^*)$ such that voter $i$ approves both $p_{\text{in}}$ and $p_{\text{out}}.$ That is, $p_{\text{in}}, p_{\text{out}} \in s^i_{g^*}.$

    Here, voter $i$'s utility function is indifferent between $p_{\text{in}}$ and $p_{\text{out}}$ for all possible outcomes. If $p_{\text{out}}$ was included in $Q_{\{g^*\}}(b^*-1),$ then it continues to be more favourable than $p_{\text{in}},$ either in marginal utility or in the tie-breaking order as per Definition~\ref{def:vote_profiles} on $V_{-i}$. Therefore  $\mathcal{P}_{\{g^*\}}^{\text{out}}(b^*)$ must be empty.\\
   \textbf{Case 3.} There exists a project $p_{\text{in}} \in \mathcal{P}_{\{g^*\}}^{\text{in}}(b^*)$ and a project $p_{\text{out}} \in \mathcal{P}_{\{g^*\}}^{\text{out}}(b^*)$ such that voter $i$ approves neither $p_{\text{in}}$ nor $p_{\text{out}}.$ That is, $p_{\text{in}}, p_{\text{out}} \notin s^i_{g^*}.$

    Same as the previous case, voter $i$'s utility function is indifferent between $p_{\text{in}}$ and $p_{\text{out}}$ for all possible outcomes, and for the same reason as above, $\mathcal{P}_{\{g^*\}}^{\text{out}}(b^*)$ must be empty.\\
   \textbf{Case 4.} There exists a project $p_{\text{in}} \in \mathcal{P}_{\{g^*\}}^{\text{in}}(b^*)$ and a project $p_{\text{out}} \in \mathcal{P}_{\{g^*\}}^{\text{out}}(b^*)$ such that voter $i$ approves $p_{\text{in}}$ but not $p_{\text{out}}.$ That is, $p_{\text{in}} \in  s^i_{g^*}$ and $p_{\text{out}} \notin s^i_{g^*}.$

    Consider the construction of $Q_{\{g^*\}}(b^*-1)$. The marginal social welfare of $p_{\text{out}}$ is at least equal to that of  $p_{\text{in}},$ and, in the case of equality, is higher in the tie-breaking order. For a group of substitutes, this means that either $p_{\text{out}}$ has a higher rank than $p_{\text{in}}$ or has an identical number of approvals but wins on tie-breaking. Since $p_{\text{out}}$ is not approved by voter $i,$ its marginal social welfare cannot become smaller than that of $p_{\text{in}}$ on adding more elements to the outcome. This is because the votes $V_{-i}$ follow Definition~\ref{def:vote_profiles}. For any outcome, it is preferable to replace $p_{\text{in}}$ by $p_{\text{out}}.$ This contradicts the definition of $Q_{\{g^*\}}(b^*)$.\\ 
  \textbf{Case 5.}  Voter $i$ approves all of  $\mathcal{P}_{\{g^*\}}^{\text{out}}(b^*)$ and none of $\mathcal{P}_{\{g^*\}}^{\text{in}}(b^*).$ 


Recall that $b^*$ is the smallest amount of funding for which $\mathcal{P}_{\{g^*\}}^{\text{out}}(b^*)$ is non-empty.
Consider the smallest amount of funding $b^\prime$ for which a project from $\mathcal{P}_{\{g^*\}}^{\text{out}}(b^*)$ is included in $Q_{\{g^*\}}(b^\prime).$  Call this project $p^*_{\text{out}}.$

By definition of  $Q_{\{g^*\}}(b^\prime)$ as the social-welfare-maximizing bundle of size $b^\prime$ from group $\{g^*\},$ the marginal social welfare (or tie-breaking order) of $p^*_{\text{out}}$ is higher than any element of $\mathcal{P}_{\{g^*\}}^{\text{in}}(b^*).$  That is, $U(Q_{\{g^*\}}(b^\prime)) \geq U(Q_{\{g^*\}}(b^\prime) \setminus p^*_{\text{out}} \cup \rho)$ for all $\rho \in \mathcal{P}_{\{g^*\}}^{\text{in}}(b^*),$ and in the case of equality, $ p^*_{\text{out}}$ is preferred in tie-breaking.

Let $\mathcal{P}_{\text{i-app}}$ denote the set of projects approved by voter $i$ that are in $Q_{\{g^*\}}(b^*-1) \setminus Q_{\{g^*\}}(b^\prime)$. 

Further, let $\mathcal{P}_{\text{stay}}$  denote the subset of $\mathcal{P}_{\text{i-app}}$ that is in the outcome $Q_{\{g^*\}}(b^*).$ 


If $\mathcal{P}_{\text{stay}}$ is empty, then the marginal utility that $p^*_{\text{out}}$ gives to voter $i$ in $Q_{\{g\}}(b) \cup p^*_{\text{out}}$ is exactly as much as it does in $Q_{\{g\}}(b^\prime).$ That is,   $u_i(Q_{\{g^*\}}(b^*) \cup p^*_{\text{out}}) - u_i(Q_{\{g^*\}}(b^*)) = u_i(Q_{\{g^*\}}(b^\prime)) - u_i(Q_{\{g^*\}}(b^\prime)\setminus p^*_{\text{out}}).$ 

Since $ p^*_{\text{out}}$ was selected in $Q_{\{g^*\}}(b^\prime)$ over any element of  $\mathcal{P}_{\{g^*\}}^{\text{in}}(b^*), p^*_{\text{out}}$ must replace any $p_{\text{in}} \in \mathcal{P}_{\{g^*\}}^{\text{in}}(b^*)$ in the outcome $Q_{\{g^*\}}(b^*)$. This contradicts the optimality of $Q_{\{g^*\}}(b^*)$.

On the other hand, if $\mathcal{P}_{\text{stay}}$ is not empty, then it is possible that $p^*_{\text{out}}$ does not give marginal utility to voter $i$ in $Q_{\{g^*\}}(b^*) \cup p^*_{\text{out}}.$ That is, it is possible that  $u_i(Q_{\{g^*\}}(b^*) \cup p^*_{\text{out}}) - u_i(Q_{\{g^*\}}(b^*)) = 0$.    

However, $p^*_{\text{out}}$ must replace any element of $\mathcal{P}_{\text{stay}}$ since $p^*_{\text{out}}$ was preferred over it in $Q_{\{g^*\}}(b^\prime)$. This replacement does not change voter $i$'s utility but makes the other voters either strictly better off or is preferred in tie-breaking (as in $Q_{\{g^*\}}(b^\prime)$). This contradicts the optimality of $Q_{\{g^*\}}(b^*)$.
\end{proof}

We now use Lemma~\ref{lem:greedy} to study incentives of strategic voting for the vote profile of Definition~\ref{def:vote_profiles}. 

\begin{theorem}\label{thm:nash_equilibrium}
Under Assumptions,~\ref{ass:1lam},~\ref{ass:group_label}, and~\ref{ass:unitcost}, and when the true vote profile is per Definition~\ref{def:vote_profiles}, truthful voting is a  Nash equilibrium in PBPI with FC-SWM-PB.
\end{theorem}
The proof is technical and is given in Appendix~\ref{sec:proof_ne_laminar}. It crucially uses Lemma~\ref{lem:greedy}. For a voter $i$ deviating from their truthful vote, we break down their vote into that for each group $g \in [r]$ and then show that for any state of their vote in the other groups, truthful voting in a group $g$ is a weakly dominant strategy as long as the other voters' votes are per Definition~\ref{def:vote_profiles}. 

\subsubsection{Implications of the results.}
Lemma~\ref{lem:greedy} and Theorem~\ref{thm:nash_equilibrium} signify that for the class of vote profiles given in Definition~\ref{def:vote_profiles}, PBPI satisfies two key desiderata of voting mechanisms -- efficient computation of outcome (via the Greedy Algorithm~\ref{alg:greedy-gf}) and an incentive for voters to vote truthfully, thereby making voting simpler. 

Theorem~\ref{thm:nash_equilibrium} has important consequences for PB. First, under the setup of this subsection, the funding constraints do not present additional complexities for voting. Therefore PB organizers and policymakers need not worry about the cognitive complexity of the mechanism when deciding whether they must impose 1-laminar funding constraints on the PB outcome. Second, since the 1-laminar funding constraints have no role in deciding the voting strategy, it can be justified to release this information as a separate document and unclutter the actual ballot. It may be possible to run the PB election even when the parameters of the funding constraints are yet to be decided by a process which doesn't depend on the PB election. 

\section{Computational Complexity of Preference Aggregation in PBPI} \label{sec:complexity}

%
Since our mechanism for PB can model relatively complex preferences for all voters, the amount of information to be processed by the preference aggregation algorithm can be substantial.
As for all voting schemes, it is important to study the computational complexity of aggregating the votes for any objective function of the social planner. 
We first make a negative observation. Due to the complex project interactions in our model, SWM-PB cannot be solved in polynomial time unless $P=NP$. 
\begin{observation} \label{thm:nphard}
\textsc{SWM-PB} is \emph{NP-hard.}
\end{observation}

A proof is in Appendix \S\ref{sec:proof_np} and is via a reduction from the maximum set coverage problem. 
This result is unsurprising since many models with project interactions in PB and MW elections face this issue \cite{jain2020,jain2020participatory,rey2020designing,izsak2017working,izsak2018committee}. 
 FC-SWM-PB is also NP-hard since it is at least as hard as SWM-PB.

For real-world voting problems, we often deal with scenarios where some problem parameters are small. 
On a positive note, we show that FC-SWM-PB can be solved in polynomial time for reasonable real-world parameters of the model. 

We first consider the case where the number of projects that must be grouped together for project interactions is small. This is expected to be the case for most real-world PB elections. We also need the technical condition that the number of bits required to specify the budget $B$ is a fixed parameter. This can naturally be true for real-world ballots where the required precision of costs and funds is not very high. For example, in a case where the total budget is $10^6$ currency units, all costs and votes can be reasonably specified in units of $10^3$ currency units. In this example, we will have $B = 1000.$

\subsection{FPT with respect to the maximum size of a group of projects and $\log(B)$}


Let $s_{\max}$ denote the maximum size of a group of projects. We will show that if $s_{\max}$ is fixed, SWM-PB and FC-SWM-PB are computationally tractable. This suggests that PB organizers must design the ballot with reasonably small groups if computational complexity is a concern.  


\begin{theorem} \label{thm:fpt}
Under Assumption~\ref{ass:group_label}, FC-SWM-PB is FPT with respect to $(\log(B), s_{\max})$.
\end{theorem}

\begin{proof}
We first give the result for SWM-PB, that is, for the special case where there are no funding constraints (except the overall budget constraint). We then extend it to FC-SWM-PB.

For intuition, observe that for any possible fund allocation to group $g,$ we can compute the social-welfare maximizing or \emph{best} subset of $z_g$ in time $O(n 2^{|z_g|})$ by doing a brute-force search over all subsets of $z_g.$ 

Let $\mathcal{U}_{G}(b)$ denote the maximum social welfare obtained on allocating funds $b$ to groups in set $G$ and the corresponding optimal set of projects by $Q_{G}(b)$.
The solution of SWM-PB is $Q_{[r]}(B)$ and its social welfare is $\mathcal{U}_{[r]}(B)$ in this notation. 
Since social welfare is additively separable across project groups, SWM-PB can be solved by the following recursion denoted by \underline{$\mathcal{R}$-FPT.}
\begin{align*} \label{eq:dp}
     \mathcal{U}_{[r]}(B)  &= \max_{ b \in [B]} ~  \mathcal{U}_{[\left \lfloor{r/2}\right \rfloor ]} (b) +  \mathcal{U}_{[r]\backslash [\left \lfloor{r/2}\right \rfloor ]} (B-b)  .\\
     b^\prime &= \argmax_{ b \in [B]} ~  \mathcal{U}_{[\left \lfloor{r/2}\right \rfloor ]} (b) +  \mathcal{U}_{[r]\backslash [\left \lfloor{r/2}\right \rfloor ]} (B-b) . \\
     Q_{[r]}(B)  &=  ~  Q_{[\left \lfloor{r/2}\right \rfloor ]} (b^\prime) \cup  Q_{[r]\backslash [\left \lfloor{r/2}\right \rfloor ]} (B-b^\prime).
\end{align*} 
We compute the base cases $\mathcal{U}_{\{g\}}(b)$ and $Q_{\{g\}}(b)$ for all $g\in [r]$ and $b \in [B]$ and store it in a table.  
Each table entry is computed in time $O(n 2^{|z_g|})$ for group $g$.  The overall computational complexity of the recursion is therefore $O(nrB \cdot 2^{s_{\max}}).$

For extending the result to FC-SWM-PB observe that, under Assumption~\ref{ass:group_label},   all projects in a group have identical labels.
 We can  modify the recursion $\mathcal{R}$-FPT to integrate the funding constraints by imposing them on collections of groups $z_g$ that together make up a set $P_l$ for any type label $l \in L$. This can be done by setting the social utility of the terms of the recursion that violate any funding constraint to $-\infty.$ 
\end{proof}

Note that we do not need the assumption on the 1-laminar structure of the funding constraints (Assumption~\ref{ass:1lam}) for this result to hold. It will hold for any arbitrary set of type labels on projects and associated minimum and maximum funding constraints as long as Assumption~\ref{ass:group_label} holds.

\subsection{FPT with respect to the number of distinct votes}

In this subsection we deviate from the general framework of PB and study the case where the number of distinct votes is small.
Often budgetary tasks are undertaken by a small number of elected representatives and each representative may be voting on behalf of a different number of voters.
This framework may also be relevant in the realm of delegation voting \cite{green2015direct}. 
Another use case would be a PB format where (an unrestricted number of) voters are asked to choose one out of a fixed set of `prototype outcomes'. Yet another use case would be when a small number of voters are queried at random to get a `quick pulse' of the people's opinions.

For simplicity, we will overload the notation and use $n$ for the number of distinct voters, and $w_i$ for the frequency or weight of the vote of voter $i$ for each $i \in [n].$ In this notation, the social utility of an outcome $Q$ is given by $\sum_{i \in [n]} w_iu_i(Q),$ which is the objective function of FC-SWM-PB here.
\begin{theorem} \label{thm:fpt-n}
Under Assumptions~\ref{ass:group_label}, and ~\ref{ass:unitcost}, FC-SWM-PB is FPT with respect to $n.$\footnote{Here B is bounded by $m$ since projects are unit cost and therefore we don't need it as a fixed parameter.}
\end{theorem}
\begin{proof}
The outer recursion is same as \underline{$\mathcal{R}$-FPT} given for Theorem~\ref{thm:fpt}. We set the social utility of the terms of the recursion that violate any funding constraint to $-\infty.$ 

Now we describe how we solve the base cases, that is, find $\mathcal{U}_{\{g\}}(b)$ and $Q_{\{g\}}(b)$ for all $g\in [r],$ $b \in [B].$
Without loss of generality, let voters $i \in [n_c]$ set $t_g^i =1$ for group $g$, i.e., express that the group is of complementary projects (we call them category C voters), and the voters $[n]\setminus [n_c]$ set $t_g^i = 0$ (we call them category S voters). Denote $n_s = n-n_c.$

We break the problem into $2^{n_c}$ `cases', each corresponding to a unique subset of category $C$ voters. For each case, we have 2 `phases'. For case $k \in [2^{n_c}]$, denote the corresponding subset of category C voters by $v_k.$

In phase 1,
all voters in $v_k$ are satisfied by allocating funds to the union of the projects they approve. If this is not possible with $b$ units of funds, set the utility of case $k$ to $-\infty.$ Let $b'$ funds be spent at this point. 

In phase 2, with the remaining $b-b'$ funds, we fund the projects that maximize the unsatisfied portion of the utility of the voters of category S. Denote the set of all subsets of voters of category S by $\mathcal{P}_S$. We divide the set of yet unfunded projects in group $g$ into $2^{n_s}$ parts, each corresponding to the projects approved by all voters corresponding to an element of $\mathcal{P}_S$. Denote this partition $\varrho = \{\rho_1, \rho_2, \ldots, \rho_{2^{n_s}}\}$. We then give a mixed integer program (MIP) with $2^{n_s}$ variables; variable $l \in [2^{n_s}]$ corresponding to the number of projects in $\rho_l$ funded in phase 2. See that the social utility objective function for category S voters over the projects funded in phase 2 is a function of these variables. By the result of Bredereck et al. \cite{bredereck2020mixed}, concave utility functions can also be incorporated in MIP, and the runtime is exponential in the number of variables.

The case with the highest social welfare is chosen as the outcome of $\mathcal{U}_{\{g\}}(b)$ and $Q_{\{g\}}(b)$. The solution is an outcome of FC-SWM-PB by construction.
The overall runtime complexity is $O(r m 2^{2^{O(n)}}).$
\end{proof}

\textbf{Discussion:} The runtime is doubly exponential in $n$, resulting from having project interaction in the objective. Due to this, the techniques of \cite{knop2020voting}, who gave singly exponential time algorithms for several combinatorial voting problems, do not directly apply to our setup. We can, however, use the approximation scheme of \cite{skowron2017fpt} to get an arbitrarily close approximation in singly exponential time (our model is `p-subseparable' per their terminology). For the case of general costs, we can get a $1-1/e$ approximation in phase 2 of our algorithm via the algorithm of \cite{sviridenko2004note} for submodular function maximization under knapsack constraints -- this results in an overall $1-1/e$ approximation for FC-SWM-PB in singly exponential time.

\section{Conclusions}
\label{sec:conclusion}
\vspace{-0.5cm}
While the use of \emph{categories} or \emph{groups} of projects in PB ballots is now standard in theory \cite{jain2020,fairstein2021proportional,jain2021partition}, experimental studies (e.g. the study on Amazon Mechanical Turk by Fairstein et al. \cite{fairstein2023participatory}), and practice (e.g. the 2020 Long Beach, USA PB election at \url{https://budget.pbstanford.org/longBeach2020}), there is no existing work on leveraging this partition to design a preference elicitation method which can enable voters to express a wide variety of project interactions with only a reasonable amount of cognitive effort\footnote{A line of work in PB studies the amount of cognitive effort  a voter must apply in a voting mechanism (see, for example, \cite{benade2018efficiency,fairstein2023participatory}). While there is a lack of consensus in the literature on the relative effort that PB mechanisms require from voters, and we have not evaluated PBPI empirically, we believe that our scheme needs only a comparable amount of effort from voters as existing methods. This is because each part of a vote in PBPI (fund allocation to groups, approval within groups, and complementarity questions) is simple and intuitive.}. We fill this gap by providing a mechanism with this property, which can also 
naturally integrate 1-laminar funding constraints without creating additional strategic or computational complexities, is deterministic\footnote{Being deterministic is often a desirable property for voting mechanisms. One of the reasons is the difficulty of verifying implementation
correctness in randomized schemes.}, is computationally tractable in reasonable parameter regimes (Theorem~\ref{thm:fpt}), and is also robust to unilateral strategic deviations for a class of vote profiles (Theorem~\ref{thm:nash_equilibrium}).
Therefore, our proposed design PBPI is a strong candidate for PB in the real world. 

Empirically studying the expressiveness and simplicity of our PB mechanism is important. This may be done similar to how Fairstein et al. \cite{fairstein2023participatory} study other common PB mechanisms. Including more complex project interactions in an intuitive voting scheme for PB continues to be an important avenue for future research. For our  positive results on strategic voting, we need to make assumptions on the cost of projects and drop the complementarity project interactions -- designing strategyproof mechanisms for PB without these assumptions would be a major contribution. Studying the properties of preference aggregation methods that maximize the Nash welfare \cite{nash1950bargaining} or characterizing the core of PB \cite{fain2016core,munagala2022auditing} under project interactions are also interesting research directions. 

\bibliographystyle{splncs04} 
\vspace{-0.5cm}
\bibliography{references}

\newpage
\appendix
\section{Appendix}

\subsection{Proof of Observation \ref{thm:nphard}: \textsc{SWM-PB} is \emph{NP-hard.}} \label{sec:proof_np}

\begin{proof}
The result is obtained via a reduction from the maximum set coverage (MSC) problem, known to be NP-hard~\cite{williamson2011design}.  An instance of the MSC problem entails a \emph{universal set} $A$ with $|A|$ elements. There is a \emph{collection} $D$ of $|D|$ subsets of $A,$ denoted by $d_1,d_2,\ldots, d_{|D|}.$ The problem is to identify $k$ sets from $D$ whose union covers $A$ maximally.

We now construct an instance of PB. Project costs are $c_j =1~ \forall j \in [m]$. There is only one group, i.e., $r = 1.$ The group is non-contradictory and all voters $i \in [n]$ set $t^i_1 = 0$ i.e., there are no complementary projects. Given an instance of MSC, we have $n=|A|$ voters, and $m = |D|$ projects. The budget is $B$ units, where $B = k.$
All voters $i \in [n]$ set $f^i_1 =1.$ Voters $i\in [n]$ approve a subset $s_1^j$ of the projects, which we construct per the instance of MSC.



%
The universal set $A$ corresponds to the set of all voters.
Let $y_j$ be the set of voters who approve project $j$, i.e., $y_j = \{i|j \in s^i_1\}.$ 
For a given instance of MSC, we construct a corresponding instance of PB with $m = |D|$ projects, and each set $d_j \in D$ maps to $y_j$. This shows that finding an outcome of SWM-PB is as hard as solving the MSC problem.
%
%
\end{proof}


\subsection{Proof of Observation~\ref{obs:not_sp_not_complementary}}
\begin{proof}
 The example below shows a profitable strategic deviation only under substitution project interaction. 

  There are $n=7$ voters, $m = 10$ projects, and the budget is $B=2$. All projects are unit-cost. There are no 1-laminar funding constraints. Projects $\{p_1,p_2,p_3\}$ form the group $z_1$, and all other projects are in singleton groups. In the truthful votes, in group $z_1$, $p_1$ is approved by voters $1,2,$ and $3$; $p_2$ is approved by voters $4,5,$ and $6$; and $p_3$ is approved by voters $2,3,4,5,$ and $7$. All voters $i\in [7]$ set the fund allocation for group $1$ to $1$, i.e., $f_1^i = 1.$ All voters $i \in [7]$ also approve and allocate funds to project $p_{i+3}.$ (That is, set $f_g^i =1$ for the singleton group containing the project $p_{i+3}$ and also add it to the corresponding approval set $s_g^i$).  The tie-breaking order prefers projects with a lower sum of indices.

Under truthful voting, the bundle $\{p_1,p_2\}$ is funded in SWM-PB (preferred in tie-breaking with a social utility of $6$).
On the other hand, if voter $7$ approves $p_4$ instead of $p_{10},$ then the bundle $\{p_3,p_4\}$ is funded in SWM-PB with a reported social utility of $7$. This outcome increases voter $7$'s utility by $1$ unit. 
\end{proof}

\subsection{Example of Vote Profile In Definition~\ref{def:gf}} \label{sec:example_vote_profile}

\begin{example} \label{ex:vote_profile}
    There are $m=9$ projects, $n=4$ voters, and $r=2$ groups with non-contradictory projects. The groups are $z_1 = \{p_1,p_2,p_3\}$ and 
    $z_2 =\{p_4,p_5,\ldots,p_9\}.$ 
    
    In group $1,$ the approval sets are $ s_1^1 = \{p_1,p_2\}, s_1^2 = \{p_2,p_3\}, s_1^3 = \{p_2\},$  and $s_1^4 = \{p_3\}.$ The fund allocations for this group are $f^1_1=f^2_1=2,$ and $f^3_1=f^4_1=1$. See that the approval sets do not satisfy any strict total order over subgroups of projects; however, since $f_1^i = |s_1^i|$ for all voters $i \in [4],$ group $1$ is of independent projects as per the votes. Therefore, the votes satisfy Definition~\ref{def:vote_profiles} for this group.

    In group 2, the approval sets are $s_2^1 = \{p_4,p_5,p_6\}, s_2^2 = \{p_4,p_5\}, s_2^3 = \{p_7,p_8,p_9\},$ and $s_2^4 = \{p_8,p_9\}.$ Their fund allocations are $f^1_2=f^2_2=1, $ and $f^3_2=f^4_2=2$. Since $f_2^i \neq |s_2^i|$ for some voters $i \in [4],$ group $2$ does not satisfy option 1 of Definition~\ref{def:vote_profiles}, i.e., it is not of independent projects.
    
    However, it satisfies option 2 of Definition~\ref{def:vote_profiles} since  $z_2$  can be partitioned into two sets with $\lambda_{2,1}=\{p_4,p_5,p_6\}$ and $\lambda_{2,2}=\{p_7,p_8,p_9\}.$
     Voters 1 and 2 only approve  projects in $\lambda_{2,1}.$ Whereas voters 3 and 4 only approve  projects in  $\lambda_{2,2}$. A strict total order of projects followed by the voters in $\lambda_{2,1}$  and $\lambda_{2,2}$ is $(p_4 \succ	 p_5 \succ p_6)$ and $(p_9 \succ p_8 \succ	p_7)$ respectively. 
     
%
\end{example}

\subsection{Proof Of Theorem~\ref{thm:nash_equilibrium}} \label{sec:proof_ne_laminar}

\begin{proof}
Let $V_{-i}$ denote the vote profile of all voters other than voter $i.$ Let voter $i's$ truthful (unobservable) vote be $(\hat{f}^i, \hat{s}^i)$. This unobserved vote is per Definition~\ref{def:vote_profiles}.  Their actual or observed vote is $(f^i,s^i)$. This may deviate from Definition~\ref{def:vote_profiles} arbitrarily. 

Say we drop the constraint of $\sum_{g \in [r]} f_g^i \leq B$ for voter $i$; we show that even without the constraint, truthful voting is a weakly dominant strategy for voter $i$ if $V_{-i}$ is per Definition~\ref{def:vote_profiles}. Note that voter $i$'s truthful vote $(\hat{f}^i, \hat{s}^i)$ satisfies the budget constraint per model definition.
When $b$ units of funds are allotted to group $g,$ recall that $\mathcal{U}_{\{g\}}(b)$ denotes the reported (i.e., observed, but not necessarily ``true'') social welfare   
of an optimal bundle of size $b$ from group $g,$ and $Q_{\{g\}}(b)$ denotes the said bundle.

For any group of projects $g \in [r]$, consider the scenario where voter $i$ sets $\{(f^i_{\rho}, s^i_{\rho}) \}_{ \rho \neq g, \rho \in [r]}$ for the other groups and is deciding their vote for group $g$, with default initialization of $(\hat f^i_g = 0, \hat s_g^i = \emptyset)$. 
In the proof, we study how voter $i$'s true utility changes relative to the outcome in this state of their vote, on setting their vote in group $g$.

We show that for any report $\{(f^i_{\rho}, s^i_{\rho}) \}_{ \rho \neq g, \rho \in [r]}$, voter $i$'s truthful vote $(\hat f_g^i,\hat s_g^i)$ on group $g$ is a weakly dominant strategy. This is sufficient to prove the theorem.

Consider the following exhaustive cases on the group and voter $i$'s true vote.


    \textbf{1. Voter $i$ derives no utility from group $g$, i.e., $\hat f_g^i = 0$:}
    
    Setting $ f_g^i > 0$ and $ s_g^i \neq \emptyset$ can only increase $\mathcal{U}_{\{g\}}(b)$ for any $b$. This can only reduce the funding from the other groups per Lemma~\ref{lem:greedy}. Per Lemma~\ref{lem:greedy}, this cannot cause the addition of any new project in the outcome $Q$ from outside of group $g$ and hence cannot increase voter $i$'s true utility. Therefore, in this case, truthful voting in group $g$ is a weakly dominant strategy.

    \textbf{2a. Group $g$ is of contradictory projects.}
    
    Recall that a voter can approve at most one project in this group. Let $p$ be the project in $\hat s_g^i.$ Since the utilities of projects in contradictory groups are additively separable, each project from group $g$ can be considered a separate singleton group for the purpose of preference aggregation.
   
    If voter $i$ votes truthfully in group $g$, they maximize the observed social welfare of project $p$. This may get project $p$ included in $Q$ if not already present. Per Lemma~\ref{lem:greedy}, exactly one project is eliminated from $Q$ if $p$ gets included. Voter $i$'s true utility can not decrease.

    There are two ways for voter $i$ to lie in contradictory groups: by setting $ f_g^i = 0$ or $( f_g^i = 1, s_g^i = \{p^\prime\})$ where $p^\prime \neq p.$  In the former, their true utility doesn't change. In the latter, they can only lose 1 unit of true utility if $p^\prime$ gets into $Q$ by replacing a project from which $i$ gets true utility. Overall, it is impossible to increase true utility by deviating from the truthful vote in groups of contradictory projects.
    %
   
    \textbf{2b. Group $g$ is a group of independent projects per Definition~\ref{def:vote_profiles}.}
    
    \underline{\textit{Case I:}} Voter $i$ sets $f_g^i = | s_g^i|.$ (see that $ f_g^i > |s_g^i|$ is identical to $ f_g^i = | s_g^i|$ for the utility function, and therefore we need not consider it separately). In this case, the observed social welfare of the projects in group $g$ are additively separable.  
    Adding a project $p$ to $s_g^i$  can cause its inclusion to $Q$ while removing exactly one project from $Q$ as per Lemma~\ref{lem:greedy}. If project $p$ is in the true vote $\hat s_g^i$ then this can only increase $i$'s true utility, and if project $p \notin \hat s_g^i$, then this can only decrease $i$'s true utility by displacing a  project beneficial for them. Therefore, setting $s_g^i = \hat s_g^i$ is a weakly dominant strategy.

    \underline{\textit{Case II:}} $f_g^i < |s_g^i|.$ For the purpose of the Greedy Algorithm~\ref{alg:greedy-gf}, this vote is equivalent to truncating $s_g^i$ to the top $f_g^i$ projects per their social utilities in $V_{-i}$ and the tie-breaking order. 
    We now have the same problem as in \textit{Case I}. Adding all of $\hat s_g^i$, and no projects from outside of it, in the ``truncated'' $s_g^i$  is, therefore, a weakly dominant strategy. This corresponds precisely to truthful voting.
        
   \textbf{2c. Group $g$ is a group of substitute projects per Definition~\ref{def:vote_profiles}.}
   
   Recall that we assume a strict total order on subgroups in groups of substitute projects per Definition~\ref{def:vote_profiles}. We use the following notation.  There is a partition $\Lambda_g = \{\lambda_{g,k}\}_{k \in [q_g]}$ of $z_g$ such that there is a strict total order within each subgroup $\lambda_{g,k}.$ Every voter in $V_{-i}$ approves projects from at most one of $\lambda_{g,k}$ per Definition~\ref{def:vote_profiles}.  Let $\lambda_{g,*}$ be the subgroup from which voter $i$ derives true utility.

    For any $f_g^i > 0$ and $s_g^i \subseteq \lambda_{g,*}$, untruthfully expanding $s_g^i$ by approving projects from $z_g\setminus\lambda_{g,*}$ can only include additional projects from $z_g\setminus\lambda_{g,*}$ to the outcome $Q$.
    This follows from Lemma~\ref{lem:greedy}. 
    This will remove an equal number of projects already in $Q$, and $i$'s true utility cannot increase. Therefore, not approving any project from $z_g\setminus\lambda_{g,*}$ is a weakly dominant strategy.

    It now remains to decide $ f_g^i$ and $s_g^i$ such that $s_g^i \subseteq \lambda_{g,*}.$

Let, in the true vote, $\hat s_g^i$ have projects from rank $1$ to $\psi$ in $\lambda_{g,*}$ (the `rank' is per the strict total order). 
Setting $\hat f_g^i \geq \psi$ is equivalent to $\hat f_g^i = \psi$ for voter $i$'s true utility, so we consider only $\hat f_g^i \leq \psi.$ 

Consider the case where voter $i$ votes truthfully in group $g.$ Let $b^*$ funding be allotted to $\lambda_{g,*}$ -- in this case, the top $b^*$ rank projects from $\lambda_{g,*}$ (up to tie-breaking) get funded. This is because, in any outcome, if we replace a lower-ranked project $p_l$ with a higher-ranked project $p_h$, the social welfare cannot decrease since, per Definition~\ref{def:vote_profiles}, all voters who get utility from $p_l$ also get the same utility from $p_h$.
In this outcome, voter $i$'s utility from group $g$ is $\min(f^i_g, b^*).$ We now show that any deviation cannot improve the utility.

We consider two cases of $b^*$ ($\geq$ or $< \hat f^i_g$) and show no profitable deviation exists in either. For a deviation $(f^i_g, s^i_g),$ let $b$ funds be allocated to $\lambda_{g,*}.$

 \underline{\textit{Case I: $b^* \geq \hat f^i_g$.}} The true utility of voter $i$ from the group is capped at $\hat f^i_g.$ 
 If $b \geq b^*,$ then the true utility from other groups can only decrease on making a deviation (per Lemma~\ref{lem:greedy}), but that from group $g$ cannot increase.
Whereas, if $b < b^*,$  the true utility of voter $i$ from group $g$ must decrease by at least $b^*-b.$ 
This is because only those projects can be removed from the outcome $Q$ whose approval is removed by voter $i$ in the deviation. 
The projects which were in $Q$ without $i$'s contribution to their marginal social utility, remain in $Q$ despite any deviation in voter $i$'s vote. Voter $i$'s true utility from outside the group can increase by at most $b^*-b$. Overall, therefore, $i$'s true utility cannot increase.

\underline{\textit{Case II: $b^* < \hat f^i_g$.}} If $b \leq b^*,$  the true utility of voter $i$ from group $g$ decreases in this deviation by at least $b^*-b.$ This is because $b$ projects have at most $b$ utility. The utility from outside the group can increase by at most $b^*-b$ per Lemma~\ref{lem:greedy}.
Also, even for $b>b^*,$ voter $i$ cannot increase their true utility from the group because their utility-producing projects from the groups have the maximum possible reported marginal social utility in the true vote, and per Lemma~\ref{lem:greedy}, they cannot be included to $Q$ at a lower marginal social utility. 
%
%
\end{proof}

\subsection{Proof of Theorem~\ref{propostion:strategyproof}} \label{sec:knapsack_lam}
\begin{proof}
Recall that in Theorem~\ref{thm:nash_equilibrium}, we showed that truthful voting is a weak Nash equilibrium for vote profiles per Definition~\ref{def:vote_profiles}.  In Theorem~\ref{propostion:strategyproof}, we consider only singleton groups, there are no project interactions, and all vote profiles are per  Definition~\ref{def:vote_profiles}. Per Theorem~\ref{thm:nash_equilibrium}, truth-telling is a weakly dominant strategy for a voter $i$ when the other voters vote truthfully. Suppose some other voters $i' \in I$ did not vote truthfully. Since all votes on singleton groups are per Definition~\ref{def:vote_profiles}, the untruthful votes can be seen as the truthful votes of another set of voters. Now it is a weakly dominant strategy for voter $i$ to vote truthfully per Theorem~\ref{thm:nash_equilibrium}. This completes the proof.
\end{proof}

\subsection{Proof of Theorem~\ref{obs:non-laminar}} \label{sec:nonlam}

\begin{proof}
Consider a non-laminar labelling $L$. There must exist two labels $l_1$ and $l_2$ such that $P_{l_1}$ is neither a subset nor a superset of $P_{l_2}$. Define $P^{only}_{l_1} = P_{l_1} \setminus P_{l_2}$ and  $P^{only}_{l_2} = P_{l_2} \setminus P_{l_1}$ and $P^{any} = P_{l_1} \cap P_{l_2}$. The cardinality of each of these sets is at least 1. Set the minimum and maximum funding constraints for labels $l_1$ and $l_2$ to 1 and set the minimum and maximum constraints for every other label to 0 and $\infty$, respectively. Let $B =2.$ We assume that the projects in $P^{any}$ have a higher preference in the tie-breaking order followed by projects in $P^{only}_{l_1}$ and projects in $P^{only}_{l_2}$.

Observe that there can be only two possible types of outcomes satisfying the funding constraints.

\begin{itemize}
    \item Only one project is funded from the set $P^{any}$.
    \item One project each is funded from the sets $P^{only}_{l_1}$ and $P^{only}_{l_2}$.
\end{itemize}

Consider projects $p_1 \in P^{only}_{l_1}$, $p_2 \in P^{only}_{l_2},$ and $p_3 \in P^{any}.$  Consider the following vote profile. Voter 1 approves projects $p_1$ and $p_3$, voter 2 approves only project $p_3$, and voter 3 approves only project $p_1$. In this case, only project $p_3$ is the outcome of FC-SWM-PB.

In this case, voter 3 has the incentive to approve both projects $p_1$ and $p_2$ in their vote. Now $\{p_1,p_2\}$ is the outcome of FC-SWM-PB, which increases the utility of voter 3 by 1 unit. 
\end{proof}



\end{document}